\documentclass{article}

\usepackage[top=0.6in, bottom=0.75in, left=0.625in, right=0.625in]{geometry}
\usepackage{subcaption}
\usepackage[unicode]{hyperref}
\usepackage{amsmath,amsthm,amssymb}
\usepackage[utf8x]{inputenc}
\usepackage[english]{babel}
\usepackage{mathtools}
\usepackage{calc}
\usepackage{xcolor}
\usepackage{url}
\usepackage[nottoc,numbib]{tocbibind}
\usepackage{listings}
\usepackage{caption}
\usepackage[section]{placeins}
\usepackage{tikz}
\usepackage{pdfpages}
\usepackage{placeins} 
\usepackage{flafter}  
\usepackage{cleveref} 
\usepackage{algorithm2e}

\usetikzlibrary{arrows.meta}
\usepackage[nottoc]{tocbibind}

\usepackage{graphicx}
\graphicspath{ {./} }

\usepackage{tocloft}
\usepackage{authblk}

\SetKw{Break}{break}

\newcommand{\RN}[1]{
\textup{\uppercase\expandafter{\romannumeral#1}}
}

\newcommand{\squad}{
    \hspace{0.5em}
}

\makeatletter
\let\@msm@th@eqref\eqref
\renewcommand{\eqref}[1]{%
  \begingroup
  \leavevmode
  \color{violet}%
  \hypersetup{linkbordercolor=[named]{violet}}%
  \@msm@th@eqref{#1}%
  \endgroup
}
\makeatother

\newtheorem{theorem}{Theorem}
\newtheorem{lemma}{Lemma}
\newtheorem{corollary}{Corollary}
\newtheorem{definition}{Definition}
\newtheorem{note}{Note}

\allowdisplaybreaks

\title{Asymptotics of the number of possible endpoints of a random walk on a directed Hamiltonian metric graph}
\author{D.V. Pyatko, V.L. Chernyshev}
\affil{\textit{National Research University Higher School of Economics (HSE University), Myasnitskaya Street, 20, Moscow, 101000, Russia}}
\date{\today}

\begin{document}

\maketitle

\begin{abstract}
    In this paper, the leading term of the asymptotics of the number of possible
final positions of a random walk on a directed Hamiltonian metric graph
is found. Consideration of such dynamical systems could be motivated by problems of propagation of narrow wave packets on metric graphs.
\end{abstract}

\textit{\textbf{Key words:} counting function, directed graph, dynamical system, Barnes—Bernoulli polynomial}

\tableofcontents

\section{Introduction and problem statement}
\label{sec1}

Let us consider a directed metric graph (one-dimensional cell complex, see the book \cite{KB} and references therein). Each of its edges is a smooth regular curve and has length, as well as permitted direction of movement. We will consider a situation in general position and assume that all edge lengths are linearly independent over the field of rational numbers. We will also fix the vertex, which we will call the starting one. A point leaves it at the initial moment of time (see \cite{Lovasz}, \cite{RW}). At each vertex with non-zero probability, we can select an outgoing edge for further movement. Reversals on the edges are prohibited. To analyze the number of possible endpoints of such a walk, it is useful to assume that all possibilities are realized. Thus, we arrive at the following dynamical system. At the initial moment of time, points begin to move with unit speed from the starting vertex along all the edges that start from it. As soon as one of the points is at the vertex, a new point appears at each incident vertex, which begins to move towards the end of the edge, and the old one disappears. If several points simultaneously come to one vertex at once, then all of them disappear, while new points appear as if one point arrived at the vertex. Our main task is to investigate $N(T)$ --- the number of moving points at time $T$ on various finite compact graphs.

This dynamical system has already been studied for the case of undirected graphs (see \cite{Chernyshev2017}, \cite{RCD}). For the number of moving points, a polynomial approximation was obtained, that is, a description of a polynomial of degree $E - 1$ was given, where $E$ is the number of edges of the graph that approximates $N(T)$ up to a certain power of the logarithm. Consideration of such dynamical systems was motivated by problems of propagation of narrow wave packets on metric graphs and hybrid manifolds (see \cite{ChSh} and references there). Quite recently, for a certain class of directed graphs (one-way Sperner graphs), the leading term of the asymptotics was obtained in the paper \cite{2021}.

In this article, we will consider arbitrary directed Hamiltonian graphs.

\subsection{
The structure of the text
and considered graphs}

\begin{definition}
By $N(T)$ we will denote the number of moving points at time moment $T$, and by $N_{x}(T)$ the number of times $\leq T$ such that at the vertex number $x$ points entered.
\end{definition}

Our main task is to find the asymptotics of $N(T)$ for the graphs, which will be described below. We will consider directed Hamiltonian graphs with $n = |V|$. We assume that the random walk starts from the vertex number $1$ and the lengths of all edges are linearly independent over $\mathbb{Q}$.
Since the graph contains a Hamiltonian cycle, we can renumber the vertices $\{2, \ldots, n\}$ so that there is a cycle $(1, 2, \ldots, n)$. Note that for this graph $N(T)$ will be the same as for the original graph at any time moment.

Recall that the first Betti number of a metric graph (that is, a one-dimensional cell complex) can be found by the formula: $\beta = |E| - |V| + 1$.
We will consider graphs with $\beta \geq 2$. For $\beta = 1$, we obtain a graph-cycle for which $N(T) = 1$ for any time $T > 0$.


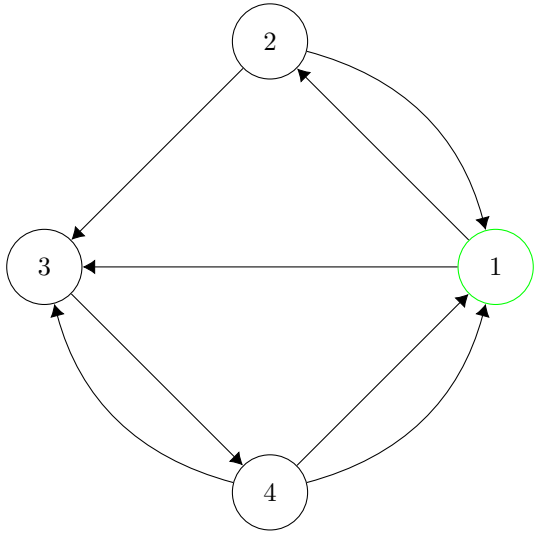
\begin{figure}[!htb]
\begin{center}
\tikzstyle{circ node}=[circle,draw=black]
\tikzstyle{arrow style}=[->,-{Latex[width=2mm]}]
\begin{tikzpicture}[minimum size=1cm, scale=3]
    \node[circ node,draw=green] (1) at (0 * 90:1cm) {$1$};
    \node[circ node] (2) at (1 * 90:1cm) {$2$};
    \node[circ node] (3) at (2 * 90:1cm) {$3$};
    \node[circ node] (4) at (3 * 90:1cm) {$4$};

    \draw[arrow style] (1) edge node {} (2);
    \draw[arrow style] (2) edge node {} (3);
    \draw[arrow style] (3) edge node {} (4);
    \draw[arrow style] (4) edge node {} (1);

    \draw[arrow style] (4) edge[bend right=30] node {} (1);
    \draw[arrow style] (2) edge[bend left=30] node {} (1);
    \draw[arrow style] (4) edge[bend left=30] node {} (3);
    \draw[arrow style] (1) edge node {} (3);
\end{tikzpicture}
\caption{An example of Hamiltonian graph with start vertex}
\label{fig:graph}
\end{center}
\end{figure}

A few words about the further arrangement of work. In Section $3.1$ we show that it suffices to represent $N_{1}(T)$ as
\begin{gather*}
N_{1}(T) = aT^{\beta} + bT^{\beta - 1} + O(T^{\beta - 2}),
\end{gather*}
where $\beta$ is the first Betti number of our complex to obtain the asymptotics of $N(T)$. Sections $2.1, 2.2, 2.3, 2.4$ are devoted to finding the required representation for the class of graphs described above.
Section $2.5$ contains statements that can be used to simplify the calculation of the coefficient. In section $3.2$ we obtain the final formula for $N(T)$ using the results from $3.1$ and $2.1, 2.2, 2.3, 2.4, 2.5$.

\section{Classification of paths to the start vertex}

\subsection{Algorithm}

Since the graph is Hamiltonian and because of the renumbering, for each vertex $i$ there is an edge to the next vertex: for $i \in \{1, \ldots, n - 1\}$ this is $i + 1$, and for $n$ it is $1$. We will call such edges inner (if there are several such edges for the vertex $i$, then we take any). A cycle of inner edges will be called an inner cycle. The rest of the edges will be called outer. A cycle with an outer edge will be called an outer cycle. For each vertex, consider all the edges outgoing from it. We fix the order on them so that the inner edge is the last one.\\

Let us consider an example.


\begin{figure}[!htb]
\begin{center}
\tikzstyle{circ node}=[circle,draw=black]
\tikzstyle{arrow style}=[->,-{Latex[width=2mm]}]
\begin{tikzpicture}[minimum size=1cm, scale=3]
    \node[circ node,draw=green] (1) at (0 * 90:1cm) {$1$};
    \node[circ node] (2) at (1 * 90:1cm) {$2$};
    \node[circ node] (3) at (2 * 90:1cm) {$3$};
    \node[circ node] (4) at (3 * 90:1cm) {$4$};

    \draw[arrow style] (1) edge[draw=red] node {} (2);
    \draw[arrow style] (2) edge[draw=red] node {} (3);
    \draw[arrow style] (3) edge[draw=red] node {} (4);
    \draw[arrow style] (4) edge[draw=red] node {} (1);

    \draw[arrow style] (4) edge[bend right=30] node {} (1);
    \draw[arrow style] (2) edge[bend left=30] node {} (1);
    \draw[arrow style] (4) edge[bend left=30] node {} (3);
    \draw[arrow style] (1) edge node {} (3);
    \begin{scope}[yshift=-1.5cm]
        \node[text width=6cm] at (0, 0) { Inner edges are {\color{red} red}, outer edges are black};
    \end{scope}
    \begin{scope}[xshift=3cm]
        \node[circ node,draw=green] (1) at (0 * 90:1cm) {$1$};
        \node[circ node] (2) at (1 * 90:1cm) {$2$};
        \node[circ node] (3) at (2 * 90:1cm) {$3$};
        \node[circ node] (4) at (3 * 90:1cm) {$4$};

        \draw[arrow style] (1) edge node[left] {{\color{red} 2}} (2);
        \draw[arrow style] (2) edge node[right] {{\color{blue} 2}} (3);
        \draw[arrow style] (3) edge node[right] {{\color{magenta} 1}} (4);
        \draw[arrow style] (4) edge node[left] {{\color{orange} 3}} (1);

        \draw[arrow style] (4) edge[bend right=30] node[right] {{\color{orange} 1}} (1);
        \draw[arrow style] (2) edge[bend left=30] node[right] {{\color{blue}1}} (1);
        \draw[arrow style] (4) edge[bend left=30] node[left] {{\color{orange} 2}} (3);
        \draw[arrow style] (1) edge node[auto] {{\color{red} 1}} (3);
        \begin{scope}[yshift=-1.5cm]
            \node[text width=6cm] at (0, 0) {A possible edge order. The last edge is an inner edge.};
        \end{scope}
    \end{scope}
\end{tikzpicture}
\caption{Inner and outer edges, edge order}
\end{center}
\end{figure}

\FloatBarrier

All inner edges make up the Hamiltonian cycle.

\begin{definition} $ $

        Let $\mathcal{H}$ be the set of all such graphs $H$ that have the same edges and vertices as in $G$, but a non-negative integer is written on each edge so that for any vertex the sum of numbers on the incoming edges is equal to the sum of the numbers on the outgoing ones.
\end{definition}

The addition of two graphs $a \in \mathcal{H}$ and $b \in \mathcal{H}$ is defined in a natural way: in the resulting graph, on each edge, the sum of the numbers of the corresponding edges of the graphs $a$ and $b$ is written.
The multiplication of $a \in \mathcal{H}$ by an integer constant $\alpha$ is defined as the multiplication of each edge in $a$ by $\alpha$.

The number corresponding to an edge can indicate how many times it can be passed along. Equality of the sum of incoming and outgoing means that each vertex was entered and left the same number of times. The set $\mathcal{H}$ will be needed further in order to describe the algorithm on graphs from this set.

Let us present an algorithm that splits $H \in \mathcal{H}$ into cycles of a certain type.
Let $H'= H$. In $H'$, the numbers on the edges and the marks will change during the running of the algorithm.
\begin{center}
    \textit{Algorithm for splitting into cycles}
\end{center}

\begin{center}
\fbox{\begin{minipage}{40em}
\texttt{Step 1. Starting from the vertex $1$, each time we select the first unmarked edge and walk along it. When we arrive at a vertex where we have already been, we let $c$ be the formed cycle.} \\

\texttt{Step 2. While in the graph $H'$ all edges along the cycle $c$ are greater than $0$, subtract $1$ along them and take the cycle $c$ in the result.} \\

\texttt{Step 3. If the cycle $c$ is inner and zero, then we stop. Otherwise, mark the first zero outer edge in the cycle $c$. Go back to step $1$.} \\
\end{minipage}}
\end{center}

\begin{figure}[!htb]
\begin{center}
\includegraphics[scale=0.2]{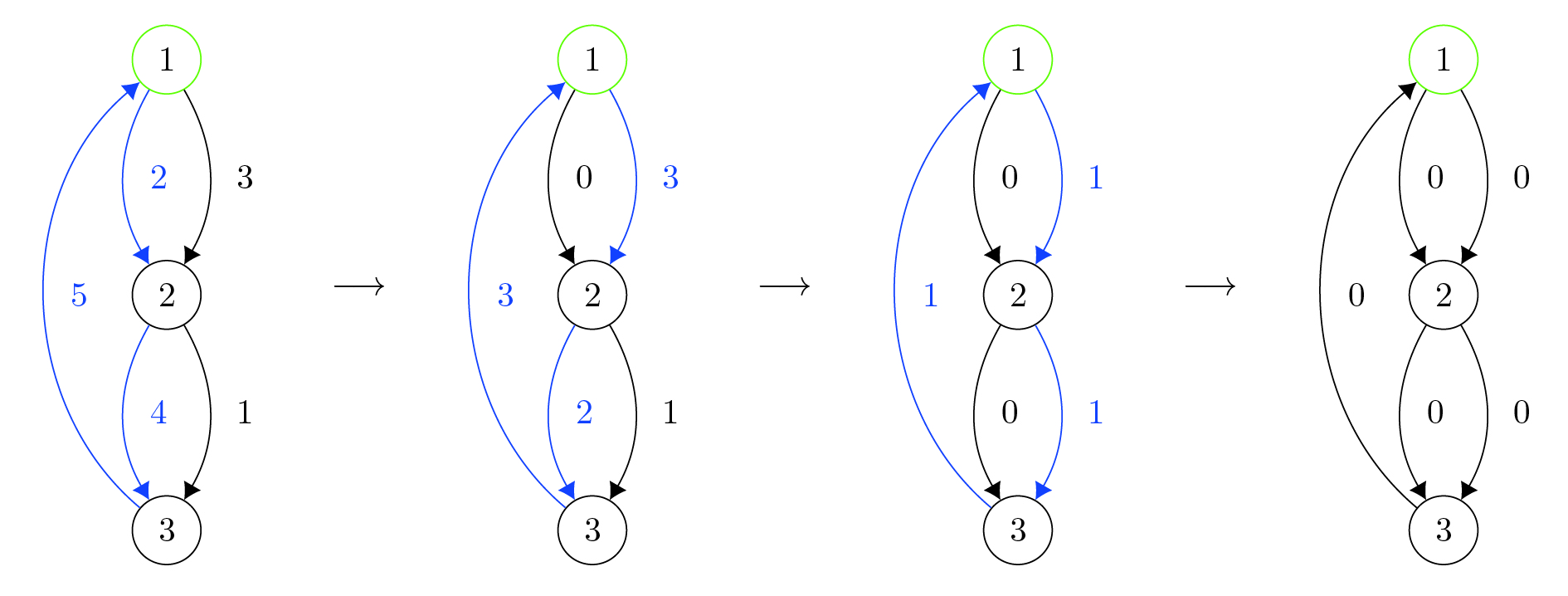}
\end{center}
\caption{An example of how the algorithm works.}
\end{figure}

The graph $H'$ changes during \texttt{step 2} when numbers on the edges are getting less and \texttt{step 3} when marks on the edges appear.

During \texttt{step 2}, the cycle $c$ is taken in the result without numbers on the edges.

\begin{lemma}[Correctness of the algorithm for dividing into cycles]
\label{lemma:algo}

Suppose we run the algorithm on $H \in \mathcal{H}$, and $c$ is any cycle that is obtained at the end of \texttt{step 2},

$H' \in \mathcal{H}$ --- what is "left" of $H$ after several steps of the algorithm.
\begin{enumerate}
\item If after \texttt{step 2} the resulting cycle $c$ is not inner or zero, then it has a zero outer edge.
\item The algorithm will split $H$ into cycles so that all edges in $H'$ at the end of the algorithm will be zero.
\end{enumerate}
\end{lemma}
\begin{proof}
To prove the assertions, we need a simple but useful observation that the algorithm retains the property that at any vertex the sum of incoming edges is equal to the sum of outgoing ones.

First, let us show that if the cycle $c$ is inner, then it is zero. Suppose this is not the case, then the cycle $c$ contains at least one
nonzero edge. Since $c$ is inner, then all outer edges are already marked, hence all outer edges are zero.
Due to the invariance of the equality of the sums of incoming and outgoing numbers on the edges, we get that nonzero numbers are written in $c$ on all edges, but this contradicts
the choice of $c$, which, according to \texttt{step 2}, must have a zero edge.

Thus, it is sufficient to consider the case when the cycle $c$ is not inner.
Let $e = (v, u)$ be an edge with zero weight in $c$. If it is outer, then everything is proven. Let's assume it's inner. Since the inner edge is the last one outgoing from the vertex, the edge $e$ is the last one from $v$, moreover, the number on $e$ is $0$, so all edges from $v$ have zero numbers (the edges up to $e$ have zero numbers because they are marked, and only edges with zero numbers are marked). We again use invariance: all edges in $v$ are also zero and, accordingly, the edge $e'= (w, v)$ of the cycle $c$ is also zero. If it is outer, then it is proven, otherwise it is inner and similar reasoning can be continued. In the reasoning, we go first from the vertex $v$ to $v - 1$, then $v - 2, \ldots, 1$, $n, \ldots, v + 1$ and show that either the next edge is inner and we continue, or it is an outer edge and we stop. If no edge turns out to be outer, then $c$ is an inner cycle, which contradicts the assumption.

The reasoning just carried out shows that $\texttt{step 3}$ is correct.
The algorithm ends due to the finitness of the edges and the appearance of more and more new marks after each step.
After the completion of the algorithm, $H'$ will be zero, since all the outer edges will be marked and therefore zero, and the inner edges that make up the inner cycle will be zero due to the stopping condition.
\end{proof}

\begin{note}
    \label{note:beta}
    For any $H \in \mathcal{H}$, there will be exactly $\beta$ iterations (\texttt{steps 1}) of the algorithm.
\end{note}
\begin{proof}

    Each iteration, except for the last one, can be bijectively associated with the outer edge, which was marked at \texttt{step 3}.
    There are $|E| - |V|$ outer edges in total, but it remains to take into account the last iteration with an inner cycle, thus there are $|E| - |V| + 1 = \beta$ iterations.
\end{proof}

\subsection{Tuples of generated cycles} $ $

When we discussed splitting $H \in \mathcal{H}$ into cycles using the algorithm, we said that we get in the answer a set consisting of some cycles $c$.
Let's understand how to describe the resulting sets of cycles. For this, we introduce a definition.

In all further notation, we assume that $c_{i}$ is a simple cycle in the graph $G$.

\begin{definition}
    \label{def:generated}
    The tuple $(c_{1}, c_{2}, \ldots, c_{k})$ is called the complete tuple of generated cycles if
    \begin{enumerate}
        \item $c_{1}$ is obtained using $\texttt{step 1}$ without marked edges.
        \item $c_{k}$ is the inner cycle.
        \item It holds that $c_{1} \xrightarrow{v_{1}} c_{2} \xrightarrow{v_{2}}c_{3}\xrightarrow{v_{3}} \ldots, \xrightarrow{v_{k - 1}}c_{k}$.
        The notation means the following: $v_{j}$ lies in the cycle $c_{j}$, $c_{j + 1}$ was obtained from $c_{j}$ by putting a mark on the first unmarked outer edge in $v_{j}$ and running $\texttt{step 1}$.
    \end{enumerate}
\end{definition}

In the course of building the set, a graph with marks is maintained in the definition \eqref{def:generated} in the same way as for the algorithm.

\begin{note}
    \label{note:beta_cycles}
    The complete tuple of reachable cycles always contains exactly $\beta$ cycles.
\end{note}
\begin{proof} $ $

    Each cycle, except for the last one, can be bijectively associated with an outer edge that was marked after passing this cycle.
    There are $|E| - |V|$ outer edges in total, but it remains to take into account the last inner cycle, thus, there are $|E| - |V| + 1 = \beta$ cycles in the complete tuple.
\end{proof}

\begin{definition} $ $
    \label{gen}
    \begin{enumerate}
        \item a tuple $(c_{1}, c_{2}, \ldots, c_{k})$ is called a tuple of generated cycles if \\
        $(c_{1}, \ldots, c_{k})$ is a subsequence of some complete tuple of generated cycles.
        \item a tuple $((c_{1}, a_{1}), (c_{2}, a_{2}), \ldots, (c_{k}, a_{k}))$ is called a tuple of weighted generated cycles if $(c_{1}, c_{2}, \ldots, c_{k})$ --- a tuple of generated cycles, $a_{i} \in \mathbb{N}$.
    \end{enumerate}
\end{definition}

\begin{definition} $ $
    \begin{enumerate}
        \item By the time of $H \in \mathcal{H}$ we mean $t(H) = \sum_{i = 1}^{k}{a_{i}t(e_{i})}$, where $a_{i}$ is the number on the $i$-th edge.
        \item Time of the tuple of weighted generated cycles\\
        $c = ((c_{1}, a_{1}), (c_{2}, a_{2}), \ldots, (c_{k}, a_{k}))$ is called $t(c) = \sum_{i = 1}^{k}{a_{i}t(c_{i})}$.
    \end{enumerate}
\end{definition}

\begin{lemma}
    \label{lemma:sigma}
    The algorithm defines a mapping from $H \in \mathcal{H}$ to a tuple of weighted generated cycles $c$ such that $t(H) = t(c)$.
\end{lemma}
\begin{proof} $ $

    Let's run the algorithm on $H$ and get a tuple of cycles

    $((c_{1}, a_{1}), (c_{2}, a_{2}), \ldots, (c_{\beta}, a_{\beta})), \squad a_{i} \in \mathbb{N} \cup \{0\}$.

    Here all cycles are taken from the corresponding iteration of the algorithm. In total, there are $\beta$ cycles due to note \eqref{note:beta}.
    By the lemma \eqref{lemma:algo} we obtain that $\sum_{i = 1}^{\beta}{a_{i}t(c_{i})} = t(H)$.
    Note that for $(c_{1}, \ldots, c_{\beta})$ all conditions of the complete tuple of generated cycles are satisfied.
    Let's drop the cycles with $a_{i} = 0$ and get a tuple of weighted generated cycles.
\end{proof}

By the lemma \eqref{lemma:sigma}, the algorithm defines a mapping that associates $H \in \mathcal{H}$ with a tuple of weighted generated cycles $c$. We will call this mapping $\sigma$.

Let $T = \sum_{i = 1}^{|E|}{a_{i}t(e_{i})}, \squad a_{i} \in \mathbb{N} \cup \{0\}$. Let's construct a graph $H$, in which $a_{i}$ is written on the edge $e_{i}$. Let us define $\omega(T) = H$.

Let us introduce the function $\mu = \sigma \circ \omega$. The $\sigma$ function "starts" the algorithm on the graph $H \in \mathcal{H}$ and produces the result. The function $\omega$ gives a graph by time, which has information about times each edge has been passed. In the definition of $\omega$, the graph exists due to the choice of $T$ on which the function $\omega$ is defined and $\omega(T)$ is uniquely defined due to the linear independence over $\mathbb{Q}$ of the edges.
It should be noted that the function $\omega$ is not defined for all times, but only for times of a special type (linear combinations of edge lengths with non-negative integer weights).

\begin{note}
    \label{note:time_inv}
    The function $\mu$ preserves time, that is $t(\mu(T)) = t(\sigma(\omega(T))) = T$, if $\omega(T) \in \mathcal{H}$ and $T = \sum_{i = 1}^{|E|}{a_{i}t(e_{i})}, \squad a_{i} \in \mathbb{N} \cup \{0\}$.
\end{note}
\begin{proof}
    It immediately follows from the fact that $t(\omega(T)) = T$ and the lemma \eqref{lemma:sigma}.
\end{proof}
\begin{definition}
    Let $\mathcal{T}$ be the set of all entry times of the starting vertex.
\end{definition}

\begin{lemma}
    \label{lemma:generated}
    For any time $T \in \mathcal{T}$, there is a tuple of weighted generated cycles $c = \mu(T)$ such that $t(c) = T$.
\end{lemma}
\begin{proof}
    Let us take $H = \omega(T)$. Since $T \in \mathcal{T}$, then $H \in \mathcal{H}$. By the lemma \eqref{lemma:sigma}, $\sigma$ is applicable to $H$. We can take $c = \mu(T)$.
    According to the note \eqref{note:time_inv} $t(c) = t(\mu(T)) = T$.
\end{proof}

\subsection{Tuples of reachable cycles} $ $

The lemma \eqref{lemma:generated} shows that we can match the time $T \in \mathcal{T}$ to a tuple of weighted generated cycles $c$ with the time $T$ using $\mu$, but the converse is not true. That is, if we take an arbitrary tuple of weighted generated cycles $c$, then $\mu^{- 1}(c)$ does not necessarily exist.

Now we want to fix this and make it so that between $\mathcal{T}$ and some set of tuples of weighted generated cycles there is a bijection specified by the function $\mu$. Note that there is already an injection. Indeed, different $T_{1} \neq T_{2}$ will give $\mu(T_{1}) \neq \mu(T_{2})$, because $t(\mu(T_{1})) = T_{1}$, $t(\mu(T_{2})) = T_{2}$ and $T_{1} \neq T_{2}$ by assumption. For bijection, it remains to obtain a surjection, but it is simple to do it: we can restrict the set of all tuples of weighted generated cycles to $\mu(\mathcal{T})$.
It turns out that there is some kind of "good" set $\mu(\mathcal{T})$ of tuples of weighted generated cycles.

Let's now understand how to calculate $N_{1}(T)$ using $\mu(\mathcal{T})$.

\begin{note}
    \label{note:time_to_reachable}
    For any time $T$ it is true that
    \begin{gather*}
        N_{1}(T) = \#\{t \in \mathcal{T}: t \leq T\} = \#\{c \in \mu(\mathcal{T}): t(c) \leq T\}.
    \end{gather*}
\end{note}
\begin{proof}
    The first equality holds by the definition of $N_{1}(T)$.
    The second holds due to the fact that $\mu$ is a bijection between $\mathcal{T}$ and $\mu(\mathcal{T})$, which preserves time.
\end{proof}

Thus, it makes sense to introduce a new type of tuples of weighted generated cycles.
For this type we will require $c \in \mu(\mathcal{T})$.
However, let's introduce this type a little differently to make it more convenient to use the definition in proofs, and then prove the equivalence to what we wrote in the previous sentence.

\begin{definition}
    \label{def:reachable_weighted}

    The tuple $c = ((c_{1}, a_{1}), (c_{2}, a_{2}), \ldots, (c_{k}, a_{k}))$ is a tuple of weighted reachable cycles if and only if
    \begin{enumerate}
         \item The tuple $c$ is a tuple of weighted generated cycles.
         \item $\mu(t(c)) = c$.
         \item $t(c) \in \mathcal{T}$.
    \end{enumerate}
\end{definition}

\begin{note} $ $
    \label{note:equiv}
    The following two conditions are equivalent: $c$ is a tuple of weighted reachable cycles $\Leftrightarrow$ $c \in \mu(\mathcal{T})$.
\end{note}
\begin{proof} $ $
    Necessity: it is enough to show $c \in \mu(\mathcal{T})$, but $c = \mu(t(c)), \squad t(c) \in \mathcal{T}$.

    Sufficiency:  $c \in \mu(\mathcal{T}) \Rightarrow c = \mu(t'), t' \in \mathcal{T}$, but $t(c) = t(\mu(t')) = t' \Rightarrow \mu(t(c)) = c, \squad t(c) \in \mathcal{T}$.

\end{proof}

The meaning of condition $2$ in the definition \eqref{def:reachable_weighted} is that the tuple $c$ will be obtained as a result of the algorithm on $\omega(t(c))$. In this sense, it is "reachable".

\begin{lemma}
    \label{lemma:induction}
    If $c = ((c_{1}, a_{1}), (c_{2}, a_{2}), \ldots, (c_{k}, a_{k}))$ --- a tuple of weighted reachable cycles, then

    \noindent
    $c' = ((c_{1}, b_{1}), (c_{2}, b_{2}), \ldots, (c_{k}, b_{k})), \squad b_{i} \in \mathbb{N}$ is also a tuple of weighted reachable cycles.
\end{lemma}
\begin{proof} $ $

    The first condition is fulfilled by the definition of a tuple of weighted generated cycles.
    Note that condition $3$ of the definition for $c$ is equivalent to the fact that $\omega(t(c))$ is an Euler graph containing the vertex $1$ if the number of passes along the edge is required to be equal to the number on the edge.
    Since $\omega(t(c))$ and $\omega(t(c'))$ have the same set of nonzero edges, then, by the theorem on the Euler property of a graph, $\omega(t(c'))$ is also Euler.

    It remains to show that the second condition is satisfied. Let $d = \mu(t(c'))$, $d = ((d_{1}, x_{1}), \ldots, (d_{k}, x_{k}))$. We need to show $c' = d$.
    Since $c$, $d$ --- are tuples of weighted generated cycles, the following representation is possible:
    \begin{align*}
         &\xrightarrow{V_{0}} (c_{1}, a_{1}) \xrightarrow{V_{1}} (c_{2}, a_{2}) \ldots (c_{k - 1}, a_{k - 1}) \xrightarrow{V_{k - 1}} (c_{k}, a_{k}), \squad V_{i} = (V_{i}^{1} \ldots V_{i}^{s_{i}}) \lor \varnothing \\
         &\xrightarrow{U_{0}} (d_{1}, x_{1}) \xrightarrow{U_{1}} (d_{2}, x_{2}) \ldots (d_{w - 1}, x_{w - 1}) \xrightarrow{U_{w - 1}} (d_{w}, x_{w}), \squad U_{i} = (U_{i}^{1} \ldots U_{i}^{p_{i}}) \lor \varnothing \\
    \end{align*}

    Here $V_{i}$ and $U_{i}$ denote a set of vertices, on the edges of which marks occurred between the cycles $c_{i - 1}$ and $c_{i}$ or $d_{i - 1}$ and $d_{i}$.

    Let the algorithm reach $(d_{i}, x_{i}), \squad i \in \{1, \ldots, w\}$ in $d$, take the cycle $d_{i}$ $x_{i}$ times and stop, then
    \begin{enumerate}
        \item $i \leq k$, $\forall j \leq i \squad d_{j} = c_{j}$, $x_{j} = b_{j}$, \\
        \item The state of the algorithm (for all vertices the first unmarked edge) is equal to the state of the algorithm if it went to $(c_{i}, a_{i})$ in $c$, took the cycle $c_{i}$ $a_{i}$ times and stopped.
    \end{enumerate}
    Let us prove this by induction. The induction base for $i = 0$: the first condition is obviously satisfied, and the second holds, since the state is initially the same.
    Let us discuss the induction transition. Let everything be satisfied for $i$, let us prove for $i + 1$:
    Let $G_{c} = \omega(t(c)), \squad G_{d} = \omega(t(d)) = \sum_{j = 1}^{w}{x_{j}d_{j}}$. By definition $d$, $G_{d} = \omega(t(c'))$. Also, let $G_{d}^{i}$ be the graph that is obtained after passing the algorithm to $(d_{i}, x_{i})$ inclusive, $G_{c}^{i}$ is defined similarly.
    Note the following:
    \begin{align*}
        G_{d}^{i} &= \sum_{j = i + 1}^{w}{x_{j}d_{j}} = \sum_{j = 1}^{w}{x_{j}d_{j}} - \sum_{j = 1}^{i}{x_{j}d_{j}} \\
                  &= \left[\text{By the induction hypothesis}\right] \sum_{j = 1}^{w}{x_{j}d_{j}} - \sum_{j = 1}^{i}{b_{j}c_{j}}  = G_{d} - \sum_{j = 1}^{i}{b_{j}c_{j}} \\
                  &= \omega(t(c')) - \sum_{j = 1}^{i}{b_{j}c_{j}} = \sum_{j = 1}^{k}{b_{j}c_{j}} - \sum_{j = 1}^{i}{b_{j}c_{j}} = \sum_{j = i + 1}^{k}{b_{j}c_{j}}. \\
    \end{align*}
    Since $G_{d}^{i} = \sum_{j = i + 1}^{k}{b_{j}c_{j}}$ contains edges with nonzero weight, then $k \geq i + 1$.
    Note that since $G_{c}^{i} = \sum_{j = i + 1}^{k}{a_{j}c_{j}}$, then in $G_{c}^{i}$ and $G_{d}^{i}$ the set of edges with nonzero weight coincides. Consider $V_{i} = (V_{i}^{1} \ldots V_{i}^{s_{i}})$ and $U_{i} = (U_{i}^{1} \ldots U_{i}^{p_{i}})$. Since the states are the same after $(c_{i}, a_{i})$ in $c$ and $(d_{i}, x_{i})$ in $d$ and the set of edges with nonzero weight is the same, then the cycles in that moment will be the same and the mark will go to the same edges, that is, $V_{i}^{1} = U_{i}^{1}$. Note that again the state is the same, we have not changed the numbers on the edges, so the set of edges with zero weight also coincides, that is $V_{i}^{2} = U_{i}^{2}$. It cannot be that $s_{i} \neq p_{i}$, because then in some cycle there would be an edge with zero weight, but in the same cycle but in a different set it would have a nonzero weight, that is $V_{i} = U_{i}$. In other words, we have come to $(c_{i + 1}, a_{i + 1})$ and $(d_{i + 1}, x_{i + 1})$. Since $V_{i} = U_{i}$, the state is the same, that is $c_{i + 1} = d_{i + 1}$.

    $G_{d}^{i} = b_{i + 1}c_{i + 1} + \sum_{j = i + 2}^{k}{b_{j}c_{j}}$. If $k = i + 1$, then $G_{d}^{i} = b_{i + 1}c_{i + 1}$ and $x_{i + 1} = b_{i + 1}$, otherwise, let us look at the edge marked by $V_{i + 1}^{1}$. It is in the cycle $c_{i + 1}$ and is used exactly $b_{i + 1}$ times. All other edges of the cycle are used $\geq b_{i + 1}$ times. That is, the cycle $c_{i + 1}$ will be subtracted exactly $b_{i + 1}$ times, $x_{i + 1} = b_{i + 1}$.

    \medskip

    We get that $k \geq w, \squad \forall j \leq w$ $c_{j} = d_{j}, \squad x_{j} = b_{j}$. That is, $d$ is a prefix of $c'$, if $k > w$,
    then $\omega(t(c')) \neq \omega(t(d))$, thus $k = w$ and $c' = d$.
\end{proof}

\begin{definition}
    \label{def:reachable}
    a tuple $(c_{1}, c_{2}, \ldots, c_{k})$ is called a tuple of reachable cycles if $((c_{1}, 1), (c_{2}, 1), \ldots, (c_{k}, 1))$ is a tuple of weighted reachable cycles.
\end{definition}

The definition \eqref{def:reachable} implies that if $(c_{1}, c_{2}, \ldots, c_{k})$ is a tuple of reachable cycles, then any tuple of weighted cycles with $(c_{1}, c_{2}, \ldots, c_{k})$ is a tuple of weighted reachable cycles and vice versa, if $(c_{1}, c_{2}, \ldots, c_{k})$ is generated but not reachable, then any tuple of weighted cycles will not be weighted reachable. This is true by the lemma \eqref{lemma:induction}.

\begin{definition}
    \label{def:reachable_length}
    Let us denote by $D_{k}$ the set of all tuples of reachable cycles with $k$ cycles.
\end{definition}

\begin{note}
    \label{note:sqcup}
    The set of all tuples of weighted reachable cycles

    can be represented in the form
    \begin{gather*}
        \bigsqcup_{i = 1}^{\beta}\bigsqcup_{d \in D_{i}}\bigsqcup_{(n_{1}, \ldots, n_{i}), n_{j} \in \mathbb{N}}((d_{1}, n_{1}), \ldots, (d_{i}, n_{i})) \\
    \end{gather*}
\end{note}
\begin{proof}
    The proof is fulfilled by the \eqref{lemma:induction} lemma and the definition \eqref{def:reachable_length}.
\end{proof}

\subsection{Formula for \texorpdfstring{$N_{1}(T)$}{}}

We are ready to formulate a statement on the exact formula for $N_{1}(T)$.

\begin{theorem}
\label{theorem:N_1}
For $N_{1}(T)$, that is, the number of time instants $\leq T$, such that points entered the vertex number $1$, the following relation is true:
    \begin{gather*}
        N_{1}(T) = \sum_{i = 1}^{\beta}{\sum_{d \in D_{i}}{\#\left\{\sum_{j = 1}^{i}{n_{j}t(d_{j}) \leq T, \squad n_{j} \in \mathbb{N}}\right\}}}.
    \end{gather*}
\end{theorem}

\begin{proof}
    \begin{align*}
        N_{1}(T) &= [\text{by definition}] \#\{t \in \mathcal{T}: t \leq T\} = \\
                 &= [\text{by note \eqref{note:time_to_reachable}}] \#\{c \in \mu(\mathcal{T}): t(c) \leq T\} \\
                 &= [\text{by note \eqref{note:equiv} and \eqref{note:sqcup}}] \\
                 &\#\left\{c \in \bigsqcup_{i = 1}^{\beta}\bigsqcup_{d \in D_{i}}\bigsqcup_{(n_{1}, \ldots, n_{i}), n_{j} \in \mathbb{N}}((d_{1}, n_{1}), \ldots, (d_{i}, n_{i})): t(c) \leq T\right\} \\
                 &= \sum_{i = 1}^{\beta}{\sum_{d \in D_{i}}{\#\left\{\sum_{j = 1}^{i}{n_{j}t(d_{j}) \leq T, \squad n_{j} \in \mathbb{N}}\right\}}}
    \end{align*}
\end{proof}

\subsection{\texorpdfstring{Tuples of weighted generated cycles of length $\beta$}{}}

\begin{definition}
    Let us define
    \begin{gather*}
        cnt_{[l;r]}^{c}(e) = \sum_{i = l}^{r}{a_{i} \cdot I\{e \in c_{i}\}},
    \end{gather*}
    where $I\{X\}$ is the truth indicator $X$, $e \in E, \squad c = ((c_{1}, a_{1}), \ldots, (c_{p}, a_{p}))$ --- a tuple of weighted reachable cycles.
\end{definition}

\begin{lemma} $ $
    \label{lemma:eq}
        If $c = ((c_{1}, a_{1}), \ldots, (c_{p}, a_{p}))$, $d = ((d_{1}, b_{1}), \ldots, (d_{k}, b_{k}))$ --- two tuples of weighted generated cycles such that $\exists i: \forall j, \squad 1 \leq j < i \squad c_{j} = d_{j} \land a_{j} = b_{j}, \squad c_{i} = d_{i} \land a_{i} \neq b_{i}$, then $t(c) \neq t(d)$.
\end{lemma}
\begin{proof} $ $
    Suppose on the contrary, $t(c) = t(d)$.
    Since the lengths of all edges are linearly independent over $\mathbb{Q}$, then for any $e \in E$ must be true that $cnt_{[1;p]}^{c}(e) = cnt_{[1;k]}^{d}(e)$. Let us look at the sequence of marks.

    Let us take $i$ from the condition. We will assume that $i < p \land i < k$ (if $i = p \land i = k$, then it is obvious, otherwise let, for example, $i = p, \squad i < k$, then in the ordered
    set $d$ after $d_{p}$ there will be a mark, and in the new unmarked edge will be different $cnt$ throughout the set).

    Let's see how the cycle $c_{i + 1}$ turned out from $c_{i}$ and $d_{i + 1}$
    from $d_{i}$. Suppose that to get $c_{i + 1}$ marks on $v_{1}, \ldots, v_{s}$, and for $d_{i + 1}$ on $u_{1}, \ldots, u_{t}$ were put.

    Let $v_{1} = u_{1}$, we denote by $e_{v}$ the edge that we marked. We get
    \begin{gather*}
        cnt_{[1;p]}^{c}(e_{v}) = cnt_{[1;i - 1]}^{c}(e_{v}) + a_{i} \\
        cnt_{[1;k]}^{d}(e_{v}) = cnt_{[1;i - 1]}^{d}(e_{v}) + b_{i},
    \end{gather*}
    but $cnt_{[1;i - 1]}^{c}(e_{v}) = cnt_{[1;i - 1]}^{d}(e_{v})$, and $a_{i} \neq b_{i}$ $\Rightarrow$ $cnt_{[1;p]}^{c}(e_{v}) \neq cnt_{[1;k]}^{d}(e_{v})$. Contradiction.

    Now let $v_{1} \neq u_{1}$. The edge marked in $c$ is denoted by $e_{v}$, and in $d$ by $e_{u}$. We get the inequalities

    \begin{equation*}
        \begin{aligned}[c]
        cnt_{[1;p]}^{c}(e_{v}) &= cnt_{[1;i - 1]}^{c}(e_{v}) + a_{i} \\
        cnt_{[1;k]}^{d}(e_{v}) &\geq cnt_{[1;i - 1]}^{d}(e_{v}) + b_{i} \\
        &\Downarrow \\
        b_{i} &\leq a_{i} \\
        \end{aligned}
        \qquad
        \begin{aligned}[c]
        cnt_{[1;k]}^{d}(e_{u}) &= cnt_{[1;i - 1]}^{d}(e_{u}) + b_{i} \\
        cnt_{[1;p]}^{c}(e_{u}) &\geq cnt_{[1;i - 1]}^{c}(e_{u}) + a_{i} \\
        &\Downarrow \\
        a_{i} &\leq b_{i} \\
        \end{aligned}
    \end{equation*}
    We get $a_{i} = b_{i}$, a contradiction.
\end{proof}

\begin{definition}
    We will call two tuples of weighted generated cycles $c = ((c_{1}, a_{1}), \ldots, (c_{m}, a_{m}))$ and $d = ((d_{1}, b_{1}), \ldots, (d_{k}, b_{k}))$ the same and denote as $c = d$, if $m = k$ and $\forall i \leq m$ $c_{i} = d_{i}$ and $a_{i} = b_{i}$.
    Otherwise, we will consider them different and denote them as $c \neq d$.
\end{definition}

\begin{lemma} $ $
        \label{lemma:beta}
        If $c = ((c_{1}, a_{1}), \ldots, (c_{\beta}, a_{\beta}))$ and $d = ((d_{1}, b_{1}), \ldots, (d_{k}, b_{k})), \squad k \leq \beta$ --- two different tuples of weighted generated cycles, then $t(c) \neq t(d)$.
\end{lemma}
\begin{proof} $ $

     For the convenience of comparing, we will assume that if there are no marks before the first cycle, then there is an "empty mark".
     Since $d$ is a tuple of weighted reachable cycles, then $d' = (d_{1}, \ldots, d_{k})$ is a tuple of generated cycles, and it is a subsequence of some complete tuple of generated cycles, then $d'$ can be written as $\xrightarrow{U_{0}} d_{1} \xrightarrow{U_{1}}d_{2}\xrightarrow{U_{2}}d_{3}\ldots d_{k - 1}\xrightarrow{U_{k - 1}}d_{k}, \squad U_{i} = (U_{i}^{1},\ldots, U_{i}^{s_{i}}) \lor \varnothing$, and the weighted tuple can be written as\\
     $\xrightarrow{U_{0}} (d_{1}, b_{1}) \xrightarrow{U_{1}}(d_{2}, b_{2})\xrightarrow{U_{2}}(d_{3}, b_{3})\ldots (d_{k - 1}, b_{k - 1})\xrightarrow{U_{k - 1}}(d_{k}, b_{k})$.

    In turn, the following representation is valid for $c$:
    \begin{gather*}
        (c_{1}, a_{1}) \xrightarrow{V_{1}}(c_{2}, a_{2}) \xrightarrow{V_{2}}(c_{3}, a_{3}) \ldots (c_{\beta - 1}, a_{\beta - 1}) \xrightarrow{V_{\beta - 1}} (c_{\beta}, a_{\beta}), \\
        \squad V_{i} = (v_{i}), \squad i \in \{1, \ldots, \beta - 1\}
    \end{gather*}

    Let's consider two cases.

    First case: $U_{0} \neq \varnothing$. Let $U_{0} = (u, \ldots)$ and $e_{u}$ --- the edge that was marked after the first passage of $u$, then $cnt_{[1; k]}^{d}(e_{u}) = 0$, $cnt_{[1; \beta]}^{c}(e_{u}) > 0$. We get $t(c) \neq t(d)$.

    Second case: $U_{0} = \varnothing$. Since $U_{0} = \varnothing$, then $\sigma = (d_{1}, b_{1}) \xrightarrow{U_{1}}(d_{2}, b_{2})\xrightarrow{U_{2}}(d_{3}, b_{3})\ldots (d_{k - 1}, b_{k - 1})\xrightarrow{U_{k - 1}}(d_{k}, b_{k})$.

    Note that views consist of two types of objects: $(\cdot, \cdot)$ and $\xrightarrow{\cdot}$. Let us go from left to right and compare for $(\cdot, \cdot)$ $a_{i}$ and $b_{i}$, and for $\xrightarrow{\cdot}$ we will compare sets $V_{i}$ and $U_{i}$ as multisets.
    If thus one is a proper prefix of another, then obviously $t(c) \neq t(d)$.
    Otherwise, there is $i: a_{i} \neq b_{i} \lor V_{i} \neq U_{i}$.
    In the first case, when $U_{i} = V_{i}$ and $a_{i} \neq b_{i}$ it can be seen that $\forall j < i$ $c_{j} = d_{j} \land a_{j} = b_{j}$, $c_{i} = d_{i} \land a_{i} \neq b_{i}$. By the lemma \eqref{lemma:eq} we get that $t(c) \neq t(d)$.
    In the second case $V_{i} \neq U_{i}$. This means that there is an outer edge $e_{u}$, which was marked in $d$, but has not yet been marked in $c$, but since $c$ is a complete tuple of weighted cycles, then at some point in the future we will switch further from this edge, but since this will happen after $\xrightarrow{V_{i}}$, then $cnt_{[i + 1;\beta]}^{c}(e_{u}) > 0$, but $cnt_{[i + 1; k]}^{d}(e_{u}) = 0$.

    \begin{gather*}
        cnt_{[1; \beta]}^{c}(e_{u}) = cnt_{[1; i]}^{c}({e_{u}}) + \underbrace{cnt_{[i + 1; \beta]}^{c}(e_{u})}_{> 0} \\
        cnt_{[1; k]}^{d}(e_{u}) = cnt_{[1; i]}^{d}(e_{u}) + \underbrace{cnt_{[i + 1; k]}^{d}(e_{u})}_{= 0} \\
        \Downarrow \\
        cnt_{[1; \beta]}^{c}(e_{u}) \neq cnt_{[1; k]}^{d}(e_{u}) \Rightarrow  t(c) \neq t(d).\\
    \end{gather*}
\end{proof}

\begin{lemma}
    \label{lemma:beta_reachable}
    Let $c = (c_{1}, \ldots, c_{\beta})$ be a tuple of generated cycles of length $\beta$,
    then $c$ is also a tuple of reachable cycles of length $\beta$.
\end{lemma}
\begin{proof}
    By the definition \eqref{def:reachable} it suffices to check that \\
    $c' = ((c_{1}, 1), \ldots, (c_{\beta}, 1))$ is a tuple of weighted reachable cycles. For it, in turn, it is necessary to check that the conditions in the definition \eqref{def:reachable_weighted} are satisfied.
    The third condition holds, since $c_{\beta}$ is an inner cycle that goes through all vertices, respectively, if we go along the edge the number of times that is written on the edge, then we get that it is an Euler graph.
    Let's check the second condition. Consider $d = \mu(t(c'))$, we have $t(c') = t(d)$. But by the lemma \eqref{lemma:beta} we get $c' = d$.
\end{proof}

\section{Calculation of the leading coefficient}

\subsection{Results for an arbitrary directed strongly connected graph \texorpdfstring{$G$}{}}

The following lemma is valid for any oriented strongly connected graph $G$ with edges linearly independent over $\mathbb{Q}$.

Let us denote by $N_{\tau, e, r}(T)$ the number of points at time moment $T$ on a segment of length $\tau$, which is at a distance $r$ from the beginning of the edge $e$.

\begin{lemma} $ $
\label{pg}
Suppose we know that the number of possible moving points that entered the first vertex by the time $T$ can be represented as: $N_{1}(T) = a_{1}T^{\beta} + b_{1}T^{\beta - 1} + O(T^{\beta - 2})$ and let $r$ be the distance to the vertex $1$ from the beginning of the segment, then:
\begin{enumerate}
    \item $N_{\tau, e, r}(T) = \tau \cdot a_{1} \cdot \beta \cdot T^{\beta - 1} + O(T^{\beta - 2}) \squad \forall e = (1, v) \in E$, $\tau > 0$
    \item $N_{\tau, e, r}(T) = \tau \cdot a_{1} \cdot \beta \cdot T^{\beta - 1} + O(T^{\beta - 2}) \squad \forall e \in E$, $\tau > 0$
    \item $N(T) = \left(\sum_{e \in E}{t(e)}\right) \cdot a_{1} \cdot \beta \cdot T^{\beta - 1} + O(T^{\beta - 2}).$
\end{enumerate}
\end{lemma}
\begin{proof} $ $
\begin{enumerate}
    \item
    If $r$ is the distance to the vertex $1$ from the beginning of the segment and $T' = T - r$, then
      \begin{gather*}
        N_{\tau, e, r}(T) = N_{1}(T') - N_{1}(T' - \tau) + O(1) =\\
        a_{1}T'^{\beta} + b_{1}T'^{\beta - 1} - (a_{1}(T' - \tau)^{\beta} + b_{1}(T' - \tau)^{\beta - 1}) + O(T'^{\beta - 2}) = \\
        = \tau a_{1} \beta T'^{\beta - 1} + O(T'^{\beta - 2}) =\tau a_{1} \beta T^{\beta - 1} + O(T^{\beta - 2}) \\
    \end{gather*}
    \item
        Let us prove first for edges of length $\tau = \varepsilon < t(e), \forall e \in E$.
        Let us put a segment on an edge $e$ of length $\varepsilon$ and a segment of length $\varepsilon$ on an edge $e_{1}$,
        which starts at $1$ and from which we can reach $e$. Let from the end of the segment on the edge $e$ to the end of the segment on the edge $e_{1}$ be a path of length $p_{1}$, and from the end of the segment on the edge $e_{1}$ to the end of the segment on the edge $e$ be the path of length $p_{2}$.
        \begin{gather*}
            N_{\varepsilon, e_{1}, r_{1}}(T - p_{2}) \leq N_{\varepsilon, e, r_{2}}(T) \leq N_{\varepsilon, e_{1}, r_{1}}(T + p_{1}) \\
            \varepsilon a_{1} \beta T^{\beta - 1} + O(T^{\beta - 2}) \leq N_{\varepsilon, e, r_{2}}(T) \leq \varepsilon a_{1} \beta T^{\beta - 1} + O(T^{\beta - 2}) \\
            \Downarrow \\
            N_{\varepsilon, e, r_{2}}(T) = \varepsilon a_{1} \beta T^{\beta - 1} + O(T^{\beta - 2})
        \end{gather*}
        Let us consider the rest of the values of $\tau$.
        We cover a segment of length $\tau$ with segments of length $\leq \min\{t(e), e \in E\}$, that is $ \{l_{i}\}_{i = 1}^{k}, \squad \sum_{i = 1}^{k}{l_{i}} = \tau, \squad l_{i}\leq \min\{t(e), e \in E\}$.
        We get:
        \begin{gather*}
            N_{\tau, e, r}(T) = \sum_{i = 1}^{k}{\left(N_{l_{i}, e, r_{i}}(T) + O(1)\right)} =\sum_{i = 1}^{k}{(l_{i} a_{1} \beta T^{\beta - 1} + O(T^{\beta - 2}))} = \\
            \left(\sum_{i = 1}^{k}{l_{i}}\right) a_{1} \beta T^{\beta - 1} + O(T^{\beta - 2}) = \tau a_{1} \beta T^{\beta - 1} + O(T^{\beta - 2}).
        \end{gather*}

    \item
        $N(T)$ is the number of points on the graph at time moment $T$, this number can be obtained as the sum of the number of points on each edge at time moment $T$. We have:
        \begin{gather*}
            N(T) = \sum_{e \in E}{N_{t(e), e, 0}(T)} = \left(\sum_{e \in E}{t(e)}\right) a_{1}  \beta  T^{\beta - 1} + O(T^{\beta - 2}) \\
        \end{gather*}
\end{enumerate}
\end{proof}

The statement just obtained implies the statement about the asymptotically equal distribution of the moving points.

\begin{corollary}[On the uniformity of distribution]
Let us know that the number of possible moving points entering the first vertex by the time $T$ can be represented as: $N_{1}(T) = a_{1}T^{\beta} + b_{1}T^{\beta - 1} + O(T^{\beta - 2})$, then the fraction of the number of moving points on an arbitrary segment of length $\tau$ will asymptotically tend to the fraction of the segment in the total length of our metric graph, namely:
    $\frac{N_{\tau, e}(T)}{N(T)} \xrightarrow{T \to +\infty} \frac{\tau}{\sum_{e \in E}{t(e)}}$
\end{corollary}

\begin{proof} $ $

    Let us use items $2$ and $3$ of the lemma \eqref{pg}.
    \begin{align*}
        \lim_{T \to +\infty}{\frac{N_{\tau, e}(T)}{N(T)}} = \lim_{T \to +\infty}{\frac{\tau a_{1} \beta + O\left(\frac{1}{T}\right)}{\left(\sum_{e \in E}{t(e)}\right) a_{1} \beta + O(\frac{1}{T})}} =\\
        &= \frac{\tau}{\sum_{e \in E}{t(e)}}. \\
    \end{align*}
\end{proof}

\subsection{Hamiltonian graphs}

Now let's return to Hamiltonian graphs.
\begin{lemma}
    \label{lemma:linear_independent}
    If $c = (c_{1}, \ldots, c_{k})$ is a tuple of generated cycles, and the graph is Hamiltonian, then $\{t(c_{1}), \ldots, t(c_{k})\}$ are linearly independent over $\mathbb{Q}$
\end{lemma}
\begin{proof} $ $

  Suppose on the contrary $\exists (\alpha_{1}, \ldots, \alpha_{k}) \neq (0, \ldots, 0)$, such that $\sum_{i = 1}^{k}{\alpha_{i}t(c_{i})} = 0$. We can assume that $\alpha_{i} \in \mathbb{Z}$. It is known that $c$ is a subsequence of a complete reachable tuple $d$ (a set of length $\beta$). Let $\gamma = \max_{i = 1}^{k}{|\alpha_{i}| + 1}$, we get
    \begin{align*}
        \sum_{i = 1}^{k}{\alpha_{i}t(c_{i})} &= 0 \\
        \sum_{i = 1}^{\beta}{(\alpha_{i}' + \gamma)t(d_{i})} &= \sum_{i = 1}^{\beta}{\gamma t(d_{i})}, \squad (\alpha_{i}' + \gamma) > 0 \\
    \end{align*}
    $\alpha'$ is obtained from $\alpha$ by extending to length $\beta$ with zeros on those positions where there is no element from $c$ in $d$.

    $(\alpha_{1}, \ldots, \alpha_{k}) \neq (0, \ldots, 0) \Rightarrow (\alpha_{1}' + \gamma, \ldots, \alpha_{\beta}' + \gamma) \neq (\gamma, \ldots, \gamma)$.
    This contradicts the lemma \eqref{lemma:beta}.
\end{proof}

We are ready to formulate the main result.

\begin{theorem}[On the asymptotics of the number of end positions of a random walk on a Hamiltonian digraph]
    \label{lemma:main}
Suppose we are given a Hamiltonian metric directed graph. Then the number of possible moving points that entered the first vertex by the time $T$ can be represented as:
$$N_{1}(T) = a_{1}T^{\beta} + b_{1}T^{\beta - 1} + O(T^{\beta - 2})$$, where $a_{1} = \sum_{d \in D_{\beta}}\frac{1}{\beta!\prod_{i = 1}^{\beta}{t(d_{i})}},$
and for the total number of moving points on the graph, the following asymptotic formula is valid:
        $$N(T) = \left(\sum_{d \in D_{\beta}}\frac{\sum_{e \in E}{t(e)}}{(\beta - 1)!\prod_{i = 1}^{\beta}{t(d_{i})}}\right) \cdot T^{\beta - 1} + O(T^{\beta - 2})$$
\end{theorem}
\begin{proof} $ $
        1. By the theorem \eqref{theorem:N_1} we have:
        \begin{align*}
            N_{1}(T) &= \sum_{i = 1}^{\beta}{\sum_{d \in D_{i}}{\#\left\{\sum_{j = 1}^{i}{n_{j}t(d_{j}) \leq T, \squad n_{j} \in \mathbb{N}}\right\}}} \\
                     &= {\sum_{d \in D_{\beta}}{\#\left\{\sum_{j = 1}^{\beta}{n_{j}t(d_{j}) \leq T, \squad n_{j} \in \mathbb{N}}\right\}}} + \sum_{i = 1}^{\beta - 1}{\sum_{d \in D_{i}}{\#\left\{\sum_{j = 1}^{i}{n_{j}t(d_{j}) \leq T, \squad n_{j} \in \mathbb{N}}\right\}}}
                             \end{align*}
            Here we apply the result of D.S. Spencer on the Barnes-Bernoulli polynomials (see \cite{RCD}) and lemma \eqref{lemma:linear_independent}. It should be noted that the required representation exists for arbitrary times of passage of edges linearly independent over $\mathbb{Q}$.\\

        \begin{align*}      &= \sum_{d \in D_{\beta}}{\frac{1}{\beta!\prod_{j = 1}^{\beta}{t(d_{j})}}} T^{\beta} + b'T^{\beta - 1} + O(T^{\beta - 2}) + b''T^{\beta - 1} + O(T^{\beta - 2}) \\
                     &= a_{1}T^{\beta} + b_{1}T^{\beta - 1} + O(T^{\beta - 2}) \\
                      \end{align*}
                      This is the representation that we assumed to exist in lemma \eqref{pg}.

        2. We apply the lemma \eqref{pg} and obtain that
        \begin{align*}
            N(T) &= \left(\sum_{e \in E}{t(e)}\right) a_{1}  \beta  T^{\beta - 1} + O(T^{\beta - 2}) =\\
             &=\left(\sum_{e \in E}{t(e)}\right) \left(\sum_{d \in D_{\beta}}\frac{1}{\beta!\prod_{i = 1}^{\beta}{t(d_{i})}}\right)  \beta  T^{\beta - 1} + O(T^{\beta - 2})= \\
            &= \left(\sum_{d \in D_{\beta}}\frac{\sum_{e \in E}{t(e)}}{(\beta - 1)!\prod_{i = 1}^{\beta}{t(d_{i})}}\right) \cdot T^{\beta - 1} + O(T^{\beta - 2})
        \end{align*}

    By the lemma \eqref{lemma:beta_reachable} one can understand $D_{\beta}$ as the set of all generated tuples of cycles of length $\beta$.
\end{proof}

\section{The discussion of the results}

\subsection{The case when there is exactly one complete reachable tuple}

In the general case, $D_{\beta}$ for a Hamiltonian graph can have many tuples of reachable cycles.
Let's consider a special case of Hamiltonian graphs for which there is exactly one complete reachable tuple.

Consider an oriented cycle. We split all the vertices into two sets. The first set will contain vertices with consecutive numbers, but will not contain the starting vertex. The second set will contain all the other vertices (the starting vertex in particular).

All edges that are outside the oriented cycle can only go from the second set of vertices to the first.


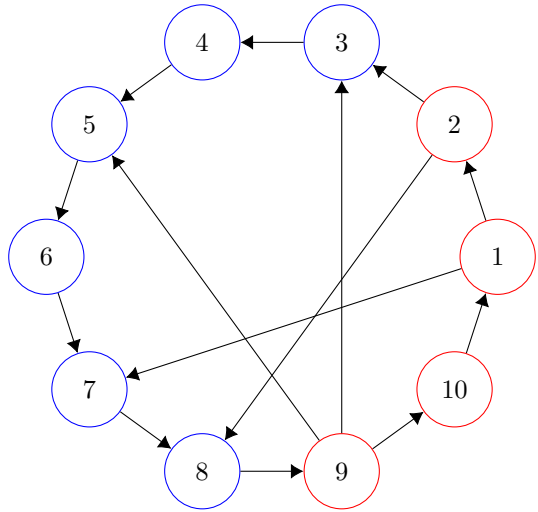
\begin{figure}[!htb]
\begin{center}
\tikzstyle{circ node}=[circle,draw=blue]
\tikzstyle{arrow style}=[->,-{Latex[width=2mm]}]
\begin{tikzpicture}[minimum size=1cm, scale=3]
    \node[circ node][draw=red] (1) at (0 * 36:1cm) {$1$};
    \node[circ node][draw=red] (2) at (1 * 36:1cm) {$2$};
    \node[circ node] (3) at (2 * 36:1cm) {$3$};
    \node[circ node] (4) at (3 * 36:1cm) {$4$};
    \node[circ node] (5) at (4 * 36:1cm) {$5$};
    \node[circ node] (6) at (5 * 36:1cm) {$6$};
    \node[circ node] (7) at (6 * 36:1cm) {$7$};
    \node[circ node] (8) at (7 * 36:1cm) {$8$};
    \node[circ node][draw=red] (9) at (8 * 36:1cm) {$9$};
    \node[circ node][draw=red] (10) at (9 * 36:1cm) {$10$};

    \draw[arrow style] (1) -- (2);
    \draw[arrow style] (2) -- (3);
    \draw[arrow style] (3) -- (4);
    \draw[arrow style] (4) -- (5);
    \draw[arrow style] (5) -- (6);
    \draw[arrow style] (6) -- (7);
    \draw[arrow style] (7) -- (8);
    \draw[arrow style] (8) -- (9);
    \draw[arrow style] (9) -- (10);
    \draw[arrow style] (10) -- (1);

    \draw[arrow style] (2) -- (8);
    \draw[arrow style] (1) -- (7);
    \draw[arrow style] (9) -- (3);
    \draw[arrow style] (9) -- (5);
\end{tikzpicture}
\caption{An example of a graph that has exactly one tuple of reachable cycles}
\end{center}
\end{figure}
{\color{blue} Blue} vertices are from the first set, {\color{red} red} are from the second.

We have exactly one complete tuple of reachable cycles $d^{0}$, because in every simple cycle the graph contains at most one outer edge.
Final formula:

\begin{gather*}
   N(T) =  \left(\frac{\sum_{e \in E}{t(e)}}{(\beta - 1)!\prod_{i = 1}^{\beta}{t(d^{0}_{i})}}\right) \cdot T^{\beta - 1} + O(T^{\beta - 2})
\end{gather*}

\subsection{Comparison with previous results} $ $

Let us now consider a class of directed graphs,
for which the leading asymptotic coefficient was obtained earlier and compare the results. These are one-way Sperner graphs, which are described in detail in \cite{2021}. The main property of such graphs was that for any vertex we had a unique path to it from the starting vertex, and the edges defining the cycles led only to the starting vertex.

By the lemma $\eqref{pg}$ it suffices to find $N_{1}(T)$. Recall that for the first Betty number, the following formula is valid: $\beta = |E| - |V| + 1$. We get that the graph contains only $|E| - (|V| - 1) = \beta$ of the "backward" edges to the start vertex.

Any time of entry to the starting vertex is given by a linear combination with nonnegative coefficients $a_{1}, \ldots, a_{\beta}$ of simple cycles $c_{1}, \ldots, c_{\beta}$, each of which contains a backward edge.
Note that if $(a_{1}, \ldots, a_{\beta}) \neq (b_{1}, \ldots, b_{\beta})$, $a_{i}, b_{i} \in \mathbb{N} \cup \{0\}$ and $a \neq 0$ or $b \neq 0$, then $\sum_{i = 1}^{\beta}a_{i}t(c_{i}) \neq \sum_{i = 1}^{\beta}b_{i}t(c_{i})$, this is true, because in each cycle $c_{i}$ there is some backward edge $e_{i}$, which is not in other cycles $c_{j}, \squad j \neq i$, and all edges are linearly independent over $\mathbb{Q}$.

We have

\begin{gather*}
    N_{1}(T) = \#\left\{ \sum_{i = 1}^{\beta}{n_{i}c_{i}} \leq T, \squad n_{i} \in \mathbb{N} \cup \{0\}\right\} = \frac{1}{\beta!\prod_{i = 1}^{\beta}{t(c_{i})}}T^{\beta} + b_{1}T^{\beta - 1} + O(T^{\beta - 2})
\end{gather*}
and the first coefficient $a_{1} = \frac{1}{\beta!\prod_{i = 1}^{\beta}{t(c_{i})}}$.

From the lemma $\eqref{pg}$ we obtain:

\begin{gather*}
    N(T) = \left(\sum_{e \in E}{t(e)}\right)a_{1}\beta T^{\beta - 1} + O(T^{\beta - 2}) = \frac{\sum_{e \in E}{t(e)}}{(\beta - 1)!\prod_{i = 1}^{\beta}{t(c_{i})}}T^{\beta - 1} + O(T^{\beta - 2})
\end{gather*}
{}
The resulting asymptotic formula for $N(T)$ coincides with the one given in \cite{2021}.

\subsection{Computer experiments}

In theorem \eqref{lemma:main} we derived the power and the leading term of the asymptotic formula for $N(T)$. To be sure that everything is correct, let's check the formula on some example. Let's take the graph from Fig. \ref{fig:graph} and add times to the edges from Graph Fig. \ref{fig:graph} in the following way:


\begin{figure}[!htb]
\begin{center}
\tikzstyle{circ node}=[circle,draw=black]
\tikzstyle{arrow style}=[->,-{Latex[width=2mm]}]
\begin{tikzpicture}[minimum size=1cm, scale=3]
    \node[circ node,draw=green] (1) at (0 * 90:1cm) {$1$};
    \node[circ node] (2) at (1 * 90:1cm) {$2$};
    \node[circ node] (3) at (2 * 90:1cm) {$3$};
    \node[circ node] (4) at (3 * 90:1cm) {$4$};

    \draw[arrow style] (1) edge node[left] {$\sqrt{3}$} (2);
    \draw[arrow style] (2) edge node[right] {$\sqrt{7}$} (3);
    \draw[arrow style] (3) edge node[right] {$\sqrt{11}$} (4);
    \draw[arrow style] (4) edge node[left] {$\sqrt{19}$} (1);

    \draw[arrow style] (4) edge[bend right=30] node[right] {$\sqrt{17}$} (1);
    \draw[arrow style] (2) edge[bend left=30] node[right] {$\sqrt{5}$} (1);
    \draw[arrow style] (4) edge[bend left=30] node[left] {$\sqrt{13}$} (3);
    \draw[arrow style] (1) edge node[auto] {$\sqrt{2}$} (3);
\end{tikzpicture}
\caption{Graph with times}
\end{center}
\end{figure}
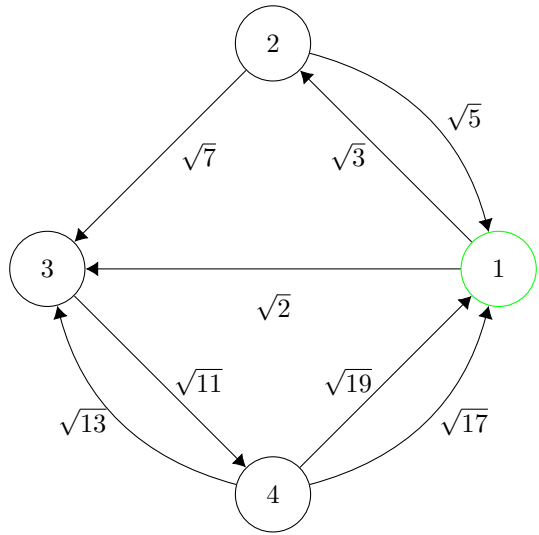
Let's notice, that all times are square roots of prime numbers, which means that they are linearly independent over $\mathbb{Q}$.
From theorem \eqref{lemma:main} it follows that:

\begin{align*}
    N(T) &= \left(\sum_{d \in D_{\beta}}\frac{\sum_{e \in E}{t(e)}}{(\beta - 1)!\prod_{i = 1}^{\beta}{t(d_{i})}}\right) \cdot T^{\beta - 1} + O(T^{\beta - 2}) \\
    \frac{N(T)}{T^{\beta - 1}} &= \sum_{d \in D_{\beta}}\frac{\sum_{e \in E}{t(e)}}{(\beta - 1)!\prod_{i = 1}^{\beta}{t(d_{i})}} + O\left(\frac{1}{T}\right) \\
\end{align*}

and

\begin{align*}
    \lim_{T \to +\infty}{\frac{N(T)}{T^{\beta - 1}}} = \sum_{d \in D_{\beta}}\frac{\sum_{e \in E}{t(e)}}{(\beta - 1)!\prod_{i = 1}^{\beta}{t(d_{i})}}
\end{align*}

Let's check that $\frac{N(T)}{T^{\beta - 1}}$ with sufficiently big $T$ is close to
$\sum_{d \in D_{\beta}}\frac{\sum_{e \in E}{t(e)}}{(\beta - 1)!\prod_{i = 1}^{\beta}{t(d_{i})}}$.

In our case $\sum_{d \in D_{\beta}}\frac{\sum_{e \in E}{t(e)}}{(\beta - 1)!\prod_{i = 1}^{\beta}{t(d_{i})}} \approx 0.000064826299$.
\begin{figure}[!htb]
\begin{center}
\includegraphics[scale=0.5]{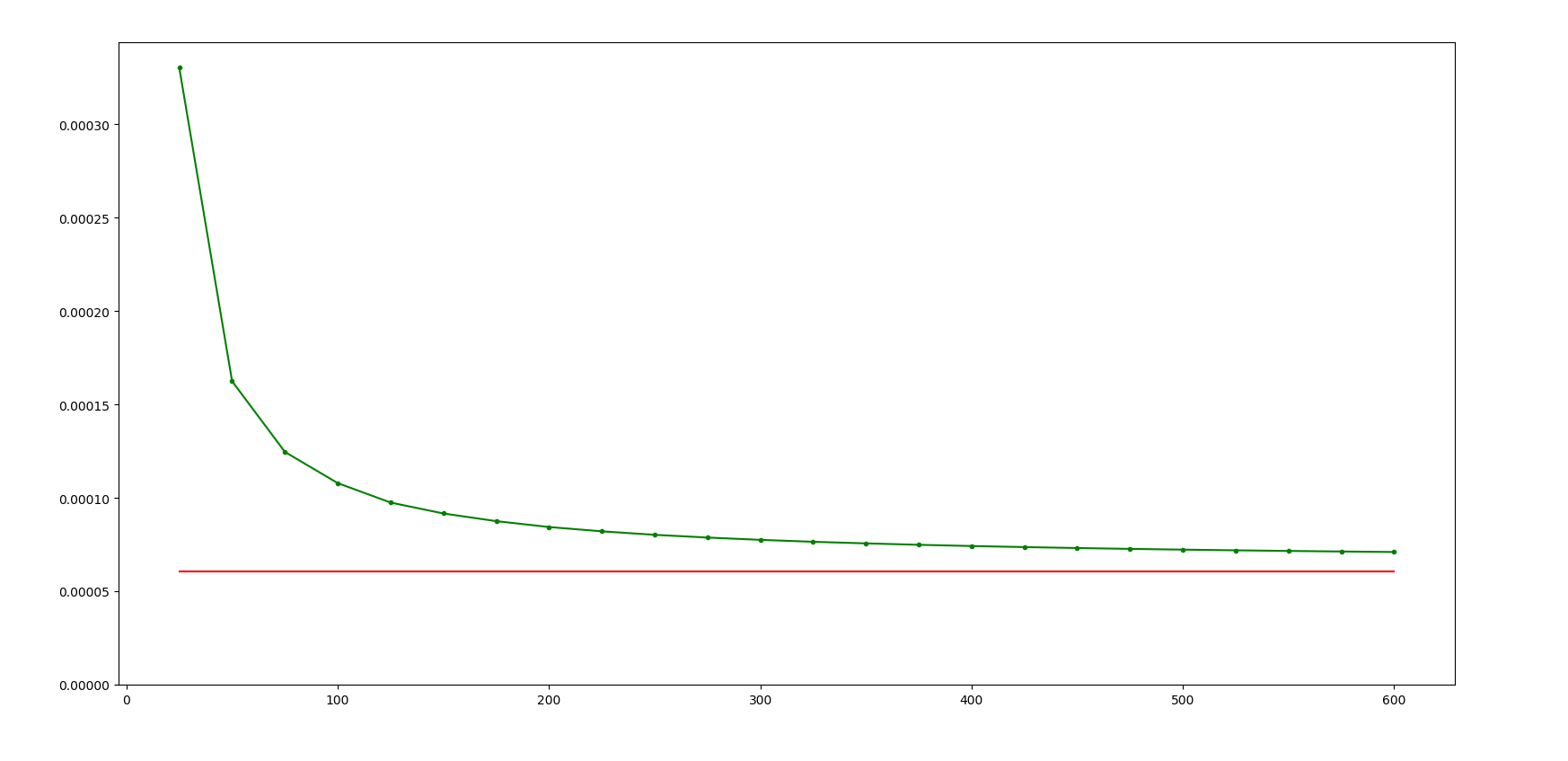}
\end{center}
\caption{Convergence of $\frac{N(T)}{T^{\beta - 1}}$}
\end{figure}

\FloatBarrier

Indeed, $\frac{N(T)}{T^{\beta - 1}}$ approaches $\sum_{d \in D_{\beta}}\frac{\sum_{e \in E}{t(e)}}{(\beta - 1)!\prod_{i = 1}^{\beta}{t(d_{i})}}$ as $T$ gets bigger.

\section{Conclusions}
We have found the asymptotics for the number of possible end positions of the random walk $N(T)$ as $T \to + \infty$ for oriented Hamiltonian graphs.
The leading coefficient is calculated by the formula $\sum_{d \in D_{\beta}}\frac{\sum_{e \in E}{t(e)}}{(\beta - 1)!\prod_{i = 1}^{\beta}{t(d_{i})}}$ and does not depend on the choice of the order on the edges, that is, $\sum_{d \in D_{\beta}}\frac{1}{\prod_{i = 1}^{\beta}{t(d_{i})}}$ is invariant under reordering on edges for Hamiltonian graphs. Understanding the nature of this invariant and studying $N(T)$ for directed strongly connected graphs that are not Hamiltonian may be a goal for further research.
\section{Acknowledgements}
D.\,V. Pyatko expresses gratitude to D.\,E. Volgin for providing computational resources for carrying out computer experiments (the results of which are consistent with the obtained theoretical calculations) and D.\,A. Polyakov for useful discussions. V.\,L. Chernyshev would like to thank A.\,A. Tolchennikov for helpful discussions.

\noindent
This study was partially supported by the RFBR grant 20-07-01103 a.

\bibliographystyle{plain}

\begin{thebibliography}{10}\label{bibliography}

\bibitem{Lovasz}
L.~Lovasz, Random walks on graphs: A survey, Combinatorics, Paul Erdos is Eighty, 1993, pp. 1-46.

\bibitem{KB}
G.~Berkolaiko, P.~Kuchment, Introduction to Quantum Graphs, Mathematical Surveys and Monographs, 2014, vol. 186 AMS.

\bibitem{RW}
V.\,L.~Chernyshev, A.\,A.~Tolchennikov, Polynomial approximation for the number of all possible endpoints of a random walk on a metric graph, Electronic Notes in Discrete Mathematics, Elsevier, 2018, vol. 70, pp. 31-35.

\bibitem{ChSh}
V.\,L.~Chernyshev, A.\,I.~Shafarevich, Statistics of Gaussian packets on metric and decorated graphs, Philosophical transactions of the Royal Society A, 2014, vol. 372, issue 2007.

\bibitem{Chernyshev2017}
V.\,L.~Chernyshev, A.\,A.~Tolchennikov, Correction to the leading term of asymptotics in the problem of counting the number of points moving on a metric tree, Russian Journal of Mathematical Physics, 2017, vol. 24, issue 3, pp. 290-298.

\bibitem{RCD}
V.\,L.~Chernyshev, A.\,A.~Tolchennikov, The Second Term in the Asymptotics for the Number of Points Moving Along a Metric Graph, Regular and Chaotic Dynamics, 2017, vol. 22, issue 8, pp. 937-948.

\bibitem{2021}
V.\,L.~Chernyshev, A.\,A.~Tolchennikov, Asymptotics of the Number of Endpoints of a Random Walk on a Certain Class of Directed Metric Graphs, Russian Journal of Mathematical Physics, 2021, vol. 28, issue 4, pp. 434-438.

\end{thebibliography}

\end{document}